\documentclass[a4paper,10pt]{elsarticle}

\makeatletter
\def\ps@pprintTitle{%
  \def\@oddhead{}%
  \def\@evenhead{}%
  \def\@oddfoot{}%
  \def\@evenfoot{}%
}
\makeatother

\usepackage{amssymb}
\usepackage{amsmath}
\usepackage{xcolor}
\usepackage{tcolorbox}
\usepackage{soul}
\usepackage{amsthm}
\usepackage{hyperref}

\theoremstyle{plain}

\newtheorem{proposition}{Proposition}

\newtheorem{theorem}{Theorem}

\usepackage{hyperref}

\theoremstyle{definition}
\newtheorem{example}{Example}
\newtheorem{definition}{Definition}

\newcommand{\inputspace}{\mathcal{R} \mathfrak{(X)}}
\newcommand{\domain}{\mathcal{R}\mathfrak{(X)}}

\newcommand{\argmax}{\mathop{\rm arg~max}\limits}

\usepackage{url}
\usepackage[T1]{fontenc}

\begin{document}

\begin{frontmatter}

\title{Mill's canons meet social ranking: A characterization of plurality}

 \author[label1]{Takahiro Suzuki\textsuperscript{[0000-0002-3436-6831]}}
  \author[label2]{Michele Aleandri\textsuperscript{[0000-0002-5177-8176]}}
   \author[label3]{Stefano Moretti\textsuperscript{[0000-0003-3627-3257]}}
   
 \affiliation[label1]{organization={Department of Civil Engineering, 
 Graduate School of Engineering, The University of Tokyo},
             addressline={Hongo Campus, 7-3-1 Hongo, Bunkyo-ku}, 
            city={Tokyo},
            postcode={113-8656}, 
            country={Japan}
            suzuki-tkenmgt@g.ecc.u-tokyo.ac.jp}

 \affiliation[label2]{organization={LUISS University},
             addressline={Viale Romania, 32},
             city={Rome},
             postcode={00197},
             country={Italy} \\
             maleandri@luiss.it}
\affiliation[label3]{organization={LAMSADE, CNRS, Universit{\'e} Paris-Dauphine, Universit{\'e} PSL},
             city={Paris},
             postcode={75016},
             country={France} \\
stefano.moretti@lamsade.dauphine.fr}

\begin{abstract}
In his book entitled ``A System of Logic, Ratiocinative and Inductive'' (1843), John Stuart Mill proposed principles of inductive reasoning in the form of five canons. To date, these canons are classic methods for causal reasoning: they are intended to single out the circumstances that are connected to the phenomenon under focus. The present paper reinterprets Mill's canons in the modern theory of social ranking solutions, which aims to estimate the power of individuals based on teams' performances. We first apply Mill's canons to determine the key success factors in cooperative performances and then characterize plurality using a strong version of Mill's first canon. Plurality is also compatible with most other canons. Thus, our results demonstrated a hidden link between classical causal reasoning and the theory of social ranking solutions.
\end{abstract}

\end{frontmatter}

\section{Introduction}\label{section:introduction}
In this paper the classic theory of induction (Mill's canons in \cite{Mill1882}) is linked and incorporated into the modern theory  of \textit{Social Ranking Solutions} (SRSs) used to determine the most contributing individual(s) in cooperative situations. 
In \cite{Mill1882}, John Stuart Mill proposed methods of induction in the form of five canons (method of agreement, method of difference, joint method of agreement and difference, method of residues, and method of concomitant variation). Each method is intended to single out the circumstances that are ``connected'' to the phenomenon under focus. Mill’s argument includes two types of inquiries: determining the cause of a given effect, and determining the effect of a given cause. In this paper, we adopt the second inquiry type (i.e., finding the causes of a phenomenon of interest). 

In contrast, the theory of SRSs aims to evaluate the contribution/power of each individual based on the performance ranking of their coalitions. This is a relatively new theory established by \cite{Moretti2017}, with many subsequent studies conducted on the topic over the last decade. In fact, many SRSs have been studied axiomatically (e.g., \textit{CP majority} \cite{Fayard2018a, Allouche2021, Suzuki2024a}, \textit{ordinal Banzhaf index} (OBI) \cite{Khani2019}, \textit{lexicographic excellence solution} (LES), their variants in \cite{Bernardi2019,Algaba2021,Aleandri2024,Suzuki2024,Suzuki2024c}, etc.). While the standard SRS is a relation-valued function (i.e., detecting a ranking of individuals), a set-valued version was also proposed in \cite{Konieczny2022} (i.e., determining the set of the most contributing individuals). This SRS is called a \textit{coalitional social choice function} (CSCF).  

This study reveals the hidden link between these two apparently independent theories. Once successful cooperation is regarded as a phenomenon, determining the individuals that contribute the most to cooperation is equivalent to determining the cause of the phenomenon. This observation enables us to consider the SRS theory as a type of causal inference; hence, we can expect that Mill's canons will shed new light on the study of SRSs. 

From this perspective, we first translate Mill's canons as axioms for CSCFs. In particular, we focus on Mill's first canon (the method of agreement), second canon (the method of difference), third canon (the joint method of agreement and difference), and fifth canon (the method of concomitant variation). We exclude the fourth canon from our main scope as it involves a more complex situation where multiple phenomena exist. This point is discussed in more detail in Section \ref{section:conclusion}. 

Our axiomatic investigation reveals that a CSCF, called \textit{plurality}, is highly connected to Mill's canons. Plurality is a new CSCF inspired by the so-called plurality in voting theory, which selects individuals who appear the most frequently in the top equivalence class. Our main results show that, when interpreted properly, plurality satisfies all of Mill's first, second, third, and fifth canon, whereas CSCFs induced by other well-known SRSs (specifically, LES and OBI) fail to satisfy some of them. This finding supports the usefulness of plurality in determining the individuals that contribute the most to successful cooperation.  

Finally, we provide additional comments on Mill's methods, which has remained influential until today. The interpretation of Mill’s methods  has been debated across many fields such as philosophy \cite{J.L.Mackie1967}, logic models \cite{Prueitt1998,Ducheyne2008a,Finn2011,Pietka2015}, and social choice theory \cite{Suzuki2023_Mill}. In \cite{Suzuki2023_Mill}, it is also conducted a comprehensive study of Mill's canons using a model close to dichotomous voting (such as Approval Voting), where each instance is either a positive/negative instance on the phenomenon.  The novelty of the proposed model lies in the ranking of instances (coalitions in SRS terminology). 

The remainder of this paper is organized as follows. Section \ref{section:model} describes the basic models, including an interpretation of Mill's methods in terms of SRSs. Our main results are presented in Section \ref{section:result}, where possible interpretations of Mill's canons are discussed and a characterization of plurality is presented. Section \ref{section:conclusion} concludes.

\section{Model}
\label{section:model}
\subsection{Basics}
Let $X$ be a finite set of individuals, with $3 \leq |X| < + \infty$.  Let $\mathfrak{X} := 2^{X}\setminus\emptyset$ denote the set of all the non empty subsets of $X$. A weak order, i.e. a reflexive, complete, and transitive binary relations\footnote{A binary relation $R$ on set $A$ is called a \textit{weak order} if
it is \textit{reflexive} (for all $a \in A$ and $aRa$), \textit{complete} (for all $a,b \in A$, we have either $aRb$ or $bRa$), and \textit{ transitive }(for all $a,b,c \in A$, if $aRb$ and $bRc$, then $aRc$)},  on set $\mathfrak{X}$ is called \emph{coalitional ranking} and the set of coalition rankings is denoted as $\mathcal{R}(\mathfrak{X})$. 
The asymmetric and symmetric
parts are denoted by $P(\succsim)$ and $I(\succsim)$, respectively, where 
$P(\succsim) := \{ (a,b)\mid (a,b) \in \succsim$  and $(b,a) \notin \succsim \}$
and $I(\succsim) := \{ (a,b)\mid (a,b) \in \succsim$  and $(b,a) \in \succsim \}$.
Given a weak order $\succsim$ the quotient order is indicated by $\succsim:\Sigma_1\succ\ldots\succ\Sigma_l$, with $l\in\mathbb{N}$, and if $l=2$ the weak order $\succsim$ is called \textit{dichotomous}. \\
For $x \in X$ and $\Delta\subseteq \mathfrak{X}$, the set of all elements in $\Delta$ containing $x$ is denoted by $\Delta[x] := \{S\in\Delta\mid x\in S\}$; the intersection of all elements in $\Delta$ is denoted as $\bigcap \Delta := \{ x \mid  x \in S,\ \forall S \in \Delta \}$. For $k\in\mathbb{N}$, define $\mathfrak{X}^{(k)}:=\{S\in\mathfrak{X}\mid \left| S \right|=k\}$.  

\begin{definition}
\label{definition:CSCF}
A \textit{coalitional social choice function} (CSCF) is a map
$F\colon\inputspace\rightarrow \mathfrak{X}$.
\end{definition}

A CSCF was studied in \cite{Konieczny2022} as a set-valued version of a SRS and we recall the definition of three specific CSCFs. To this purpose, we first need some preliminary notation. Given 
$\succsim\in\inputspace$ with quotient order $\succsim:\Sigma_{1}\succ\cdots\succ\Sigma_{l}$, and $x\in X$, let $\theta_\succsim(x):=(x_{1},x_{2},\cdots, x_{l})$, where $x_{k}:=\left| \Sigma_{k}[x]\right|$. We define the lexicographic relation $\ge^L$ between any pair of vectors $\theta_\succsim(x)$ and $\theta_\succsim(y)$ as follows: $\theta_\succsim(x) \ge^L \theta_\succsim(y)$  if, either  $\theta_\succsim(x)=\theta_\succsim(y)$, or it exists $t$ such that $x_s=y_s$ for all $s=1,\dots, t-1$ and  $x_t>y_t$.
Moreover, for $x\in X$ and $\succsim\in\inputspace$, let $u_{x}^{+,\succsim}:=|\{S\subseteq X\setminus \{x\}\mid S\cup\{x\}\succ S\}|$, $u_{x}^{-,\succsim}:=|\{S\subseteq X\setminus \{x\}\mid S\succ S\cup\{x\}\}|$, and then $s_{x}^{\succsim}:=u_{x}^{+,\succsim}-u_{x}^{-,\succsim}$.

\begin{definition}
\label{definition:example_of_CSCFs}
For $\succsim\in\inputspace$ with  $\succsim:\Sigma_{1}\succ\cdots\succ\Sigma_{l}$, 
\begin{itemize}
\item \textit{Plurality} $F^P\colon\inputspace\rightarrow \mathfrak{X}$ is a CSCF such that 
$$F^{P}(\succsim) :=\argmax _{x\in X}\, x_{1};$$
\item \textit{Lexicographic excellence solution (LES)} $F^{L}\colon\inputspace\rightarrow \mathfrak{X}$ is a CSCF such that $$F^{L}(\succsim):= \{x\in X\mid \theta_{\succsim}(x) \ge^{L} \theta_{\succsim}(y) \text{\ for all\ } y \in X \}; $$
\item  \textit{Ordinal Banzhaf Index} (OBI) $F^{OBI}:\inputspace\rightarrow\mathfrak{X}$ is a CSCF such that $$F^{OBI}(\succsim):=\argmax_{x\in X} s_{x}^{\succsim}.$$
\end{itemize}
\end{definition}

The Plurality solution assigns each coalitional ranking to the set of players that appear most frequently in the first equivalence class. This is analogous to the popular voting rule called plurality, which selects candidates who are placed first by the greatest number of voters. To the best of our knowledge, this is the first time that plurality has been referred to in the study of SRSs (including CSCFs). The Lexicographic Excellence Solution \cite{Bernardi2019}  selects, among the players chosen by the plurality solution, the player who appears most frequently in the second equivalence class. In the case of a tie, the third equivalence class is considered, and so on, until all ties are broken and one single player is selected, or some ties still persist in the last equivalence class so that multiple players are selected. The Ordinal Banzhaf Index \cite{Haret2019a} compares the number of times a player improves a coalition's ranking minus the number of times their inclusion worsens the ranking.

\subsection{Mill's canons}
\label{subsection:Mill's_canons}

The situation considered in Mill's Canons is as follows: an analyst aims at estimating the cause of a phenomenon based on the given instances. Each instance is either positive (the phenomenon occurs) or negative (the phenomenon does not occur)\footnote{In the fifth canon, however, the extent of the phenomenon will be argued.}. Furthermore, each instance is equipped with certain circumstances (i.e., some circumstances hold under the instance and others do not). The aim is to identify which circumstance(s) is considered to be the cause of the phenomenon.

\begin{example}
\label{example:SRSs_and_Mill's_canons}
In our approach, we interpret Mill's canons within the terminology of the SRS coalitional framework using the following semantics (the lefthand side represents the terminology in Mill's canons, and the righthand side represents its interpretation in cooperative contexts): 

\begin{align*}
\text{phenomenon} &\leftrightarrow \text{successful cooperation} \\
\text{an instance} &\leftrightarrow \text{a coalition} \\
\text{a circumstance} &\leftrightarrow \text{an individual}.
\end{align*}

For example, following the multiple causation framework studied in \cite{ferey2016multiple}, the set $X$ (representing the circumstances) could denote tortfeasors (individuals) involved in a loss (the phenomenon) suffered by a victim. A (instance of)  potential damage to a victim can be attributed to alternative groups of tortfeasors (coalitions) who succeed in cooperating against the victim. 
In this way, a CSCF is identified with the induction rule in Mill's sense; that is, the method with which one determines the cause (key individuals) of the phenomenon (successful cooperation).  
\end{example}

\begin{example}\label{ex:frame}
Inspired by a toy example introduced in \cite{ferey2016overdetermined}, consider a set $X=\{1,2,3\}$ of three firms polluting a river by releasing hazardous chemicals. A threshold $\alpha$ exists, beyond which the chemical concentration becomes lethal to the fish. Suppose that owing to the quantities of emitted chemicals, no firm alone is sufficient to lethally pollute the river (above $\alpha$). However, we assume that the joint emissions of chemicals by firms $1$ and $2$, as well as by $1$ and $3$, are greater than $\alpha$, whereas the joint emissions of firms $2$ and $3$ remain below $\alpha$. 
We say that coalitions  $\{1,2\}$, $\{1,3\}$ and $\{1,2,3\}$ are successful in cooperating to produce the (lethal) phenomenon of pollution because their joint action is larger than the predefined $\alpha$, whereas all the other subsets of $X$ are not successful in cooperating. 
\end{example}

We now introduce Mill's first canon (p.482 in \cite{Mill1882}) as follows. 

\begin{quote}
\textit{\textbf{First canon}}   If two or more instances of the phenomenon under investigation have only one circumstance in common, the circumstance in which alone all the instances agree is the cause (or effect) of the given phenomenon.
\end{quote}

The first canon describes the \textit{ method of agreement}. According to this method, circumstance $x$ is considered the cause of the phenomenon if $x$ is a unique common circumstance for all positive instances.  Following Example \ref{ex:frame}, firm $1$ should be considered the causally liable firm, being at the intersection of all successful coalitions.

Mill's second canon (p.483 in \cite{Mill1882}) states

\begin{quote}
\textbf{\textit{Second canon}}  If an instance in which the phenomenon under investigation occurs, and an instance in which it does not occur, have every circumstance in common save one, that one occurring only in the former; the circumstance in which alone the two instances differ is the effect, or the cause, or an indispensable part of the cause, of the phenomenon.
\end{quote}

The second canon describes the \textit{ method of difference}. If a positive instance and a negative instance differ only in the existence of circumstance $x$ ($x$ holds in the former and does not hold in the latter), then $x$ is considered the cause of the phenomenon (in other words, the second canon describes the principle of the controlled experiment).
 
\begin{example}   
Consider the scenario from Example \ref{ex:frame} without firm $3$, leaving us with $X=\{1,2\}$. In this case, the only successful coalition that produces lethal amounts of chemicals consists of firms $\{1, 2\}$. According to Mill's second canon, firms 1 and 2 should be held equally responsible for causing harm.
\end{example}

Mill's third canon is as follows:  

\begin{quote}
\textbf{\textit{Third canon}}   If two or more instances in which the phenomenon occurs have only one circumstance in common, while two or more instances in which it does not occur have nothing in common save the absence of that circumstance, the circumstance in which alone the two sets of instances differ is the effect, or the cause, or an indispensable part of the cause, of the phenomenon. 
\end{quote}

Mill emphasizes the importance of the second canon (method of difference) for causal inference. However, Mill is also aware of the difficulty of applying it, especially when circumstances cannot be selected on one's own. Hence, a third canon is introduced to expand the scope of the method of difference by indirectly applying the method of agreement. Nevertheless, we did not use this canon in the present study because of its antecedents. Because it applies the method of agreement in the first stage (saying, "If two or more instances in which the phenomenon occurs have only one circumstance in common"), it can be applied only when the first canon is applied. Hence, logically, the third canon should be weaker than the first. This will be shown later based on our formal definition of the canons (the same holds in the dichotomous model in \cite{Suzuki2023_Mill}). 

For now, we skip the fourth canon because it applies to a different case where multiple phenomena are considered. We briefly discuss this at the end of Subsection \ref{subsection:absolute-top-interpretation}. 

Finally, we proceed to the fifth canon using the concomitant variation method. It is worth noting that we implicitly assumed that the characteristics of each instance were dichotomous: they were either positive or negative. As Mill notes, there are some phenomena in which such a dichotomous view is inappropriate. 

\begin{quote}
There remains a class of laws which it is impracticable
to ascertain by any of the three methods which I have attempted to characterize, namely the laws of those Permanent Causes, or indestructible natural agents, for which it is impossible either to exclude or to isolate. That is, we can neither hinder these laws from being present, nor contrive that they shall be present alone. (\cite{Mill1882}, pp.491-492)
\end{quote}

To address these phenomena, Mill added the following canon:

\begin{quote}
\textbf{\textit{Fifth canon}}   Whatever phenomenon varies in any manner whenever another phenomenon varies in some particular manner, is either a cause or an effect of that phenomenon, or is connected with it through some fact of causation.
\end{quote}

The fifth canon focuses on concomitant relationships between circumstances and phenomena. In our context, this means that if the extent of successful cooperation improves whenever $x$ attends it, then $x$ should be the cause of success. A more detailed discussion of the contextualization and formalization of the canons introduced earlier will be presented in the next section.

\subsection{Mill's canons as axioms for SRSs}
We demonstrate that Mill's canons can be understood properly in the context of SRSs. A remaining task, however, is the interpretation of "successful cooperation.” Among the keywords shown in Example \ref{example:SRSs_and_Mill's_canons}, "successful cooperation" is not formally defined in the model of SRSs. To understand Mill's canons as formal axioms of SRSs, we need to give a specific interpretation of "successful cooperation,” or a success in the cooperative contexts.  

One straightforward interpretation of success, which we call {\it relative interpretation}, yields a better result than some coalition $S_{0}$ (which is a reference point). For each coalition $S\in\mathfrak{X}$, $S$ is considered a positive (negative) instance of the phenomenon if $S\succ S_{0}$ $\left(S_{0}\succsim S\right)$. According to the relative-interpretation, first, second, and third canons are formally stated in the following as the properties of relative-agreement, relative-difference, and relative-joint method of agreement and difference, respectively.

\begin{definition}
The CSCF $F$ is said to satify 
\begin{itemize}
\item \textit{Relative-agreement} (RAG) if, for any $\succsim\in\inputspace$ and $S_{0}\in\mathfrak{X}$, if $\bigcap_{S\succ S_{0}} S$ is a singleton, then $F\left(\succsim\right)=\bigcap_{S\succ S_{0}} S$.  
\item \textit{Relative-difference} (RDF) if, for any $\succsim\in\inputspace$, $S_{0}\in\mathfrak{X}$, and $x\in X$, if $S\cup\{x\}\succ S_{0}\succsim S$ for all $S\subseteq X\setminus\{x\}$, then we have that $F\left(\succsim\right)=\{x\}$.
\item \textit{Relative-joint method of agreement and difference} (RJAD) if, for any $\succsim\in\domain$ and $S_{0}\in\mathfrak{X}$, if $\bigcap_{S\succ S_{0}}=\{x\}$, $\bigcap_{S_{0}\succsim S}S=\emptyset$, and $x\notin \bigcup_{S_{0}\succsim S} S$, then we have that $F(\succsim)=\{x\}$. 

\end{itemize}
\end{definition}
As it will be demonstrated in Proposition \ref{proposition:relative-interpretation}, the RAG, RDF, and RJAD properties are not logically independent. To circumvent this issue, we resort to an alternative interpretation called the {\it absolute top interpretation}. According to this interpretation, this phenomenon is associated with the best results.
Coalition $S$ is considered a positive instance of the phenomenon if and only if it belongs to the best equivalence class with respect to $\succsim$. According to the absolute-top-interpretation, Mill's first, second, and third canons are represented by the following properties of top-agreement (and its stronger version), top-difference, and top-joint method of agreement and difference, respectively. 

\begin{definition}
The CSCF $F$ is said to satisfy 
\begin{itemize}
\item \textit{Top-Agreement} (TAG) if, for any $ \succsim \in \inputspace$ with $\succsim:\Sigma_{1}\succ\cdots\succ\Sigma_{l}$, $\bigcap \Sigma_{1}$ is a singleton, then $F(\succsim)=\bigcap \Sigma_{1}$. 
\item \textit{Strong Top-Agreement} (STAG) if, for any $\succsim \in \inputspace$ with $\succsim:\Sigma_{1}\succ\cdots\succ\Sigma_{l}$,  $\bigcap \Sigma_{1}\neq \emptyset$, then $F(\succsim)=\bigcap \Sigma_{1}$. 
\item \textit{Top-Difference} (TDF) if, for any $\succsim\in\inputspace$ with $\succsim:\Sigma_{1}\succ\cdots\succ\Sigma_{l}$ and $x\in X$, $\mathfrak{X} \left[x\right]\subseteq \Sigma_{1}$ and $\mathfrak{X}\left[x\right]^c\subseteq \Sigma_{2}\cup\cdots\cup\Sigma_{l}$, then $F(\succsim)=\{x\}$. 
\item \textit{Top-joint method of agreement and difference} (TJAD) if, for any $\succsim \in \inputspace$ with $\succsim:\Sigma_{1}\succ\cdots\succ\Sigma_{l}$, $\bigcap \Sigma_{1}=\{x\}$, $\bigcap_{2\leq k} \Sigma_{k}=\emptyset$, and $\bigcup_{2\leq k} \Sigma_{k}\subseteq X\setminus\{x\}$, then it follows that $F(\succsim)=\{x\}$.
\end{itemize}
\end{definition}


Finally, we formulate the fifth canon. For any $\succsim\in\inputspace$, let $C_\succsim := \{x\in X\mid S\cup\{x\}\succ S\ \text{for all}\ S\subseteq X\setminus \{x\}\}$. 

\begin{definition}
The CSCF $F$ is said to satisfy \textit{Concomitant Variation} (CV) if for any $\succsim \in \inputspace$, we have $C_\succsim \subseteq F(\succsim)$.
\end{definition}

For $k\in\mathbb{N}$, let $\left[k\right]:= \{a\in\mathbb{N}\mid a\leq k\}$ (the set of natural numbers with maximum $k$). For $\succsim,\succsim'\in\inputspace$ and $S\in\mathfrak{X}$, we say that $\succsim'$ is an \textit{$S$-deterioration of $\succsim$} if it holds that 
\begin{itemize}
    \item[(i)] $\succsim \mid_{ \{S\}^c} = \succsim'\mid_{\{S\}^c}\,$, where $\{S\}^c=\mathfrak{X}\setminus \{S\}$ and $\succsim\mid_{\Delta}$ is the weak order restricted to the elements of $\Delta\subseteq \mathfrak{X}$,
    \item[(ii)] for any $T\in\mathfrak{X}$, $\left[ T\sim S \Rightarrow T\succsim'S\right]$ and $[T\succ S  \Rightarrow T  \succ' S]$.
\end{itemize}
Finally, two additional axioms are introduced. 

\begin{definition}
\label{definition:supplementary_axioms}
The CSCF $F$ is said to satisfy 
\begin{itemize} 
\item \textit{Slide Independence} (SI) if, for any $\succsim:\Sigma_{1}\succ\cdots\succ\Sigma_{l}$, $x,y \in X$, $k_{1},k_{2}\in\left[l\right]$, $\Gamma \subsetneq\Sigma_{k_{1}}$ with $\left|\Gamma[x]\right|=\left|\Gamma[y]\right|$, and $\succsim':\Sigma_{1}'\succ'\cdots\succ'\Sigma_{l}'$, where $\Sigma_{k_{1}}'=\Sigma_{k_{1}}\setminus\Gamma$, $\Sigma_{k_2}'=\Sigma_{k_2}\cup\Gamma$, and $\Sigma_k'=\Sigma_k$ for all $k\neq k_{1},k_{2}$,  then, it follows that 
\[ \left[\ F(\succsim)\cap\{x,y\}\neq \emptyset \mbox{ and }F(\succsim')\cap\{x,y\}\neq \emptyset \Rightarrow  F(\succsim)\cap\{x,y\}=F(\succsim')\cap\{x,y\}\ \right].\]
\item \textit{Downward Monotonicity} (DMON) if,  for any $x \in X$, $S \subseteq X\setminus\{x\}$, and $\succsim,\succsim'\in\inputspace$, assuming that $\succsim'$ is an  $S$-deterioration of $\succsim$,  then, it follows that 
\[ \left[\ x\in F(\succsim)\Rightarrow x\in F\left(\succsim'\right)\ \right]. 
\]
\end{itemize}
\end{definition}

The rationale behind these two axioms is as follows: SI says that the shift of $\Gamma$ (which is a proper subset of $\Sigma_{k_{1}}$) from some equivalence class $\Sigma_{k_{1}}$ to $\Sigma_{k_{2}}$ does not affect the relative ranking between $x$ and $y$, whenever $x$ and $y$ appears the same number of times in $\Gamma$ (i.e., the number of coalitions in $\Gamma$ including $x$ and the number of coalitions in $\Gamma$ including $y$ are the same). The DMON states that the degradation of coalitions that do not include $x$ does not degrade the evaluation of $x$ (a similar axiom called down-monotonicity is found in the literature on social choice functions \cite{Taylor2002,Brandt2016a}).  

\begin{example}\label{ex:si}
Consider the scenario with three firms described in Example \ref{ex:frame}, but now suppose that a lowering of the threshold $\alpha$ due to a change of environmental conditions makes the emissions of coalitions $\{1\}$ and $\{2,3\}$ successful to produce the phenomenon of pollution. Thus, in the new ranking, $\succsim'$ coalitions $\{1\}$ and $\{2,3\}$, together with $\{1,2,3\}$, $\{1,2\}$ and $\{1,3\}$ are successful in cooperation, whereas the remaining coalitions are not. Note that the new ranking $\succsim'$ can be seen as a switch of coalitions in $\Gamma=\{\{1\}, \{2,3\}\}$  from $\Sigma_2$ to $\Sigma_1$ with respect to the ranking $\succsim$ representing the scenario in Example \ref{ex:frame}. Moreover, we have that $|\Gamma[1]|=|\{1\}|=1=|\{2,3\}|=|\Gamma[2]|$. If we assume that $F(\succsim)\cap\{1,2\}=\{1\}$ (firm $1$ is considered causally liable in $\succsim$) and $F(\succsim')\cap\{1,2\}\neq \emptyset$, then the SI property requires that $F(\succsim')\cap\{1,2\}=\{1\}$, meaning that the balanced presence of firms $1$ and $2$ in the switch should not affect the determination of causal responsibility established prior to the switch.
\end{example}

\begin{example}\label{ex:dmon}
Consider again the scenario with three firms described in Example \ref{ex:frame}, where new scientific evidence indicates that only the action of firm $2$ is safe for the environment, whereas the actions of coalitions $\{2,3\}, \{1\}$, and $\{3\}$, despite staying below the threshold, pose potential hazards to the environment. This new situation can be represented by ranking $\succsim’ \in \inputspace$ as follows:
$$ \{1,2,3\}\sim '\{1,2\} \sim’\{1,3\} \succ’ \{2,3\} \sim’ \{1\}\sim’ \{3\}\succ’ \{2\}.$$
Clearly, $\succsim’$ reflects the deterioration of coalition $\{2\}$ in the ranking $\succsim$ in Example \ref{ex:frame}, suggesting a reduced causal impact for firm $2$. Therefore, this deterioration should not affect the liability considerations for firms deemed responsible for the harm prior to new scientific evidence. Consistently, property DMON requires that firm $1$ remains causally liable in $\succsim'$ if it is considered so in the original ranking $\succsim$.   
\end{example}

\section{Results}
\label{section:result}

\subsection{On the inconsistency in the relative-interpretation axioms}

We show that Mill's canons are incompatible under a relative interpretation.

\begin{proposition}
\label{proposition:relative-interpretation}
The following holds: 
\begin{enumerate}[i)]
\item \label{item:RDF and CV are incompatible}RDF and CV are incompatible; that is, there is no CSCF satisfying both axioms.
\item  \label{item:RAG implies RDF and RJAD}RAG implies RDF and RJAD.
\item \label{item:RAG and CV are incompatible}RAG and CV are incompatible. 
\end{enumerate}
\end{proposition}

\begin{proof}\ \\

Proof of \ref{item:RDF and CV are incompatible}) Let $x,y\in X$ and $\succsim\in\domain$ be such that $\succsim: \Sigma_{1}\succ\Sigma_{2}\succ\Sigma_{3}\succ\Sigma_{4}$, where $\Sigma_{1}:=\mathfrak{X}[x]\cap\mathfrak{X}[y]$, $\Sigma_{2}:=\mathfrak{X}[x]\setminus\mathfrak{X}[y]$, $\Sigma_{3}:=\mathfrak{X}[x]^c\cap \mathfrak{X}[y]$, and $\Sigma_{4}:=\mathfrak{X}[x]^c\setminus \mathfrak{X}[y]$. Then, RDF implies that $F(\succsim)=\{x\}$ (by  $S_{0}:=\{y\}\in \Sigma_{3}$). However, we also have that $y\in C_{\succsim}$. Hence, CV demands that $y\in F(\succsim)$. This means that no CSCF can satisfy RDF and CV simultaneously.

Proof of \ref{item:RAG implies RDF and RJAD}) 
[RAG $\Rightarrow$ RJAD] is straightforward by definition. We show that the RAG implies RDF. Let $\succsim\in\inputspace$, $S_{0}\in\mathfrak{X}$, and $x\in X$ such that 
\begin{equation}
\label{equation:DF}
S\cup\{x\}\succ S_{0}\succsim S, \text{ for all } S\subseteq X\setminus\{x\}.
\end{equation} 
We first show that $x \notin S_0$. Suppose on the contrary that $x\in S_{0}$. Then, for $S_{1}:=S_{0}\setminus\{x\}$, and the relation $S_{1}\cup\{x\}=S_{0}\succ S_{0}$ is false. This contradicts equation (\ref{equation:DF}). Therefore, $x\notin S_{0}$, then any $S\succ S_{0}$ for all $S\in\mathfrak{X}$ with $x\in S$. In this special case, we have $\{x\}\succ S_{0}$. This implies that $\left(X\setminus\{x\}\right)\cap\left(\bigcap_{T\succ S_{0}} T\right)=\emptyset$. Furthermore, according to equation (\ref{equation:DF}), $x\in\bigcap_{T\succ S_{0}} T$. Therefore, $\bigcap_{T\succ S_{0}} T=\{x\}$. Using the RAG, we obtain $F(\succsim)=\{x\}$. So, we have proved that $F$ satisfies RDF.

Proof of \ref{item:RAG and CV are incompatible}) is straightforward from the above two items.\end{proof}

Proposition \ref{proposition:relative-interpretation} concisely demonstrates the logical relationship between Mill's first, second, third, and fifth canons under relative-interpretation; the item \ref{item:RAG and CV are incompatible}) and \ref{item:RDF and CV are incompatible}) prove that the fifth canon (CV) is not compatible with either the first canon (RAG) or the second canon (RDF). This indicates the difficulty of understanding Mill's canons consistently under this interpretation.

A possible escape route from this incompatibility is to weaken the definition of RAG and RDF. For instance, let us say that a CSCF $F$ satisfies  \textit{weak relative agreement } (WRAG) if, for any $\succsim\in\domain$ and for any $S_{0}\in\mathfrak{X}$, such that $\bigcap_{S\succ S_{0}}S$ is a singleton, then $F(\succsim)\supseteq\bigcap_{S\succ S_{0}}S$. WRAG is obtained from the RAG by substituting $F(\succsim)=\bigcap_{S\succ S_{0}} S$ with $F(\succsim)\supseteq\bigcap_{S\succ S_{0}} S$. Hence, WRAG is logically weaker than the RAG. In fact, a constant CSCF $F^{constX}$ such that $F^{constX}(\succsim)=X$ for all $\succsim\in\domain$ satisfies both WRAG and CV.

Proposition \ref{proposition:relative-interpretation} indicates that the relative interpretation of Mill's canons, despite their intuitive appeal, does not work. This observation led us to consider other interpretations. Subsection \ref{subsection:absolute-top-interpretation} demonstrates that the absolute-top-interpretation is successful in the sense that we can find CSCFs that satisfy Mill's first, second, and fifth canons together. 

\subsection{Characterization of Plurality in the absolute-top-interpretation}
\label{subsection:absolute-top-interpretation}
The previous subsection proved the inconsistencies observed in the relative interpretations. Here, we will prove the possibility of an absolute top interpretation. 

\begin{theorem}
\label{theorem: characterization of plurality by STAG}
A CSCF $F$ satisfies STAG, DMON, and SI if and only if it is plurality $F^{P}$.
\end{theorem}

\begin{proof}
\emph{The 'if' part:} 

\vspace{.5\baselineskip}
\noindent
\textbf{\textit{On STAG}}. By construction $F^{P}$ clearly satisfies STAG. 

\vspace{.5\baselineskip}
\noindent
\textbf{\textit{On SI}}. Let $\succsim,\succsim'\in\inputspace$ with $\succsim:\Sigma_{1}\succ\cdots\succ\Sigma_{l}$ and $\succsim':\Sigma_{1}'\succ'\cdots\succ'\Sigma_{l}'$, $k_{1},k_{2}\in\left[l\right]$, and $\Gamma\subsetneq\Sigma_{k_{1}}$ such that $\Sigma_{k_{1}}'=\Sigma_{k_{1}}\setminus\Gamma$, $\Sigma_{k_{2}}'=\Sigma_{k_{2}}\cup\Gamma$, and $\left|\Gamma[x]\right|=\left|\Gamma[y]\right|=\gamma$. For any $z\in X$, let $\theta_{\succsim}(z):=(z_{1},z_{2},\cdots,z_{l})$. 
Since $\left|\Gamma[x]\right|=\left|\Gamma[y]\right|$, for each $z=x,y$, we have that $z_{k_1}'=z_{k_{1}}-\gamma$, $z_{k_{2}}'=z_{k_{2}}+\gamma$, and $z_{k}'=z_k$, for all $k\neq k_{1},k_{2}$. Hence, we have that $x_{1}'-y_{1}'=x_{1}-y_{1}$. Using this equation, we obtain 
\begin{align*}
&\text{if } F^{P}(\succsim)\cap\{x,y\}=\{x\}, \text{then we have } x_{1}> y_{1}, \text{and so } x_{1}'> y_{1}';\\
&\text{if } F^{P}(\succsim)\cap\{x,y\}=\{y\}, \text{then we have } y_{1}> x_{1}, \text{and so } y_{1}'> x_{1}';\\
&\text{if } F^{P}(\succsim)\cap\{x,y\}=\{x,y\}, \text{then we have } x_{1}= y_{1}, \text{and so } x_{1}'= y_{1}'.
\end{align*}
In each case, if $F^{P}(\succsim)\cap\{x,y\}\neq\emptyset$, $F^{P}(\succsim')\cap\{x,y\}=F^{P}(\succsim)\cap\{x,y\}$. This proves that $F^{P}$ satisfies the SI. 

\vspace{.5\baselineskip}
\noindent
\textbf{On DMON.} Let $x\in X$, $S\subseteq X\setminus\{x\}$, and $\succsim,\succsim'\in\domain$ such that $\succsim'$ is an $S$-deterioration of $\succsim$. Let $\succsim:\Sigma_{1}\succ\cdots\cdots\succ\Sigma_{l}$ and $\succsim':\Sigma_{1}'\succ'\cdots\succ'\Sigma_{l'}'$. For each $z\in X$, let $\theta_{\succsim}(z):=(z_{1},\cdots,z_{l})$ and $\theta_{\succsim'}(z):=(z_{1}',\cdots,z_{l'}')$. 

Assume that  $x\in F^{P}(\succsim)$. This implies that $x\in\argmax_{z\in X} \ z_{1}$. By the definition of $\succsim'$, $x_{1}=x_{1}'$ and $z_{1}\geq z_{1}'$ for all $z\neq x$. Thus, it follows that $x\in \argmax_{z\in X}\ z_{1}'$, implying that $x\in F^{P}(\succsim')$. Therefore, $F^{P}$ satisfies DMON.

\vspace{.5\baselineskip}
\noindent
\emph{The 'only if' part:} Take $\succsim\in\inputspace$. We prove that (i) $F(\succsim)\subseteq F^{P}(\succsim)$ and (ii) $F^{P}(\succsim)\subseteq F(\succsim)$. \\
\noindent

Proof of (i). We prove that, for any $x,y\in X$, if $x_{1}>y_{1} \Rightarrow y\notin F(\succsim)$. Once this has been demonstrated, it follows that $\left(F^P(\succsim)\right)^c\subseteq \left(F(\succsim)\right)^c$, and therefore $F(\succsim)\subseteq F^{P}(\succsim)$.  

Suppose, on the contrary, that $x_{1} >y_{1}$ and $y\in F(\succsim)$. Let $\Sigma_{1}^{x}, \Sigma_{1}^{y}, \Sigma_{1}^{x,y},\text{ and }\Sigma_{1}^{\emptyset}$ be as follows: 
\begin{align*}
&\Sigma_{1}^{x}:=\{S\in \Sigma_{1}\mid S\cap \{x,y\}=\{x\}\},\\
&\Sigma_{1}^{y}:=\{S\in \Sigma_{1}\mid S\cap \{x,y\}=\{y\}\},\\
&\Sigma_{1}^{x,y}:=\{S\in \Sigma_{1}\mid S\cap \{x,y\}=\{x,y\}\},\\
&\Sigma_{1}^{\emptyset}:=\{S\in \Sigma_{1}\mid S\cap \{x,y\}=\emptyset\}.
\end{align*}
Let $\succsim':\Sigma_{1}'\succ'\cdots\succ'\Sigma_{l}'$, where $\Sigma_{1}':=\Sigma_{1}\setminus\Sigma_{1}^{\emptyset}$, $\Sigma_{2}':=\Sigma_{2}\cup\Sigma_{1}^{\emptyset}$, and $\Sigma_{k}':=\Sigma_k$ for all $k\geq3$. From $y\in F(\succsim)$ (the assumption) and DMON, we have $y\in F(\succsim')$. 

Note that $x_1=\left|\Sigma_{1}^{x}\right|+\left|\Sigma_{1}^{x,y}\right|$ and $y_1=\left|\Sigma_{1}^{y}\right|+\left|\Sigma_{1}^{x,y}\right|$. Thus, $x_{1}>y_{1}$ implies that $\left|\Sigma_{1}^{x}\right|>\left|\Sigma_{1}^{y}\right|\geq 0$. Therefore, there exists $\Gamma \subseteq \Sigma_{1}^{x}$ such that $\left|\Gamma\right|=\left|\Sigma_{1}^{y}\right|$. It follows that $\Sigma_{1}^{x}\setminus\Gamma\neq \emptyset$. 

Now, let $\succsim'':\Sigma_{1}''\succ''\cdots\succ''\Sigma_{l}''$, where $\Sigma_{1}'':=\Sigma_{1}'\setminus\left(\Gamma\cup\Sigma_{1}^{y}\right)$, $\Sigma_{2}'':=\Sigma_{2}'\cup\left(\Gamma\cup\Sigma_{1}^{y}\right)$, and $\Sigma_{k}'':=\Sigma_{k}'$ for all $k\geq 3$. Then, it follows that $x\in \bigcap \Sigma_{1}''$ and $y\notin \bigcap \Sigma_{1}''$ with the proof of the latter is as follows. From the definitions of $\succsim',\succsim'', \text{ and } \Gamma$, we have $\Sigma_{1}''=\Sigma_{1}^{x,y}\cup\left(\Sigma_{1}^{x}\setminus\Gamma\right)$. Since $\Sigma_{1}^{x}\setminus\Gamma\neq \emptyset$, we have that $y\notin \bigcap \Sigma_{1}''$. Therefore, STAG implies that $x\in F(\succsim'')\not\ni y$. To sum up, we have obtained that $F(\succsim')\cap \{x,y\}\neq \emptyset$ (because $y\in F(\succsim')$), $F(\succsim'')\cap \{x,y\}\neq \emptyset$ (because $x\in F(\succsim'')$), and yet $y\in F(\succsim')\setminus F(\succsim'')$. This contradicts SI.  \\

Proof of (ii): Let $y\in F^{P} (\succsim)$ and $x\in F(\succsim)$. From (i), it follows that $x\in F\left(\succsim\right)\cap F^{P}\left(\succsim\right)$, and by definition, $x_1=y_1$. There are two cases: Case 1) $\{x,y\}\in \Sigma_{1}$ and Case 2) $\{x,y\}\notin \Sigma_{1}$. 

Case 1) Let $\succsim':\Sigma_{1}'\succ'\cdots\succ'\Sigma_{l}'$, where $\Sigma_{1}':=\{\{x,y\}\}$, $\Sigma_{2}':=\Sigma_{2}\cup\left(\Sigma_{1}\setminus\Sigma_{1}'\right)$, and $\Sigma_k':=\Sigma_k$ for all $k\geq3$. Then, the STAG requires $F(\succsim')=\{x,y\}$. Recall that $F(\succsim)\neq \emptyset$ because $x\in F(\succsim)$. Therefore, the SI implies that $F(\succsim)\cap \{x,y\}=F(\succsim')\cap \{x,y\}$. Accordingly, we obtain $y\in F(\succsim)$. 

Case 2) Let $\succsim':\Sigma_{1}'\succ'\cdots\succ'\Sigma_{l}'$, where $\Sigma_{1}':=\Sigma_{1}\cup\{x,y\}$ and $\Sigma_{k}':=\Sigma_{k}\setminus\{\{x,y\}\}$ for all $k\geq2$. Subsequently, it follows that $F^{P}(\succsim')=\{x,y\}$. Hence, by (i) $F(\succsim')\subseteq \{x,y\}$, and then $F(\succsim')\cap \{x,y\}\neq \emptyset$. Therefore, the SI implies that 
\begin{equation}
\label{equation:case 2}
F(\succsim)\cap \{x,y\}=F(\succsim')\cap \{x,y\}.
\end{equation}
In summary, we verified that $x\in F(\succsim')$ (because of Equation \eqref{equation:case 2} and $x\in F(\succsim)$ by assumption),   $y\in F^{P}(\succsim')$, and $\{x,y\}\in \Sigma_{1}'$. In Case 1), $y\in F(\succsim')$ and equation \eqref{equation:case 2} guarantees that $y\in F(\succsim)$. 
\end{proof}

We introduce two new CSCFs to prove the independence of the axioms in Theorem \ref{theorem: characterization of plurality by STAG}. 
The CSCF $F^{splitP}$ (called \textit{split-plurality}) is defined, for any $\succsim\in\inputspace$, as 
$$F^{splitP}(\succsim):= \argmax_{x\in X} \ \tilde{x}_{1},$$ where $\tilde{x}_{1}:= \sum_{\substack{S\subseteq \mathfrak{X}[x]\\ S\in\Sigma_1}} \frac{1}{\left| S \right|}$. 

The CSCF $F^*$ is defined, for any $\succsim\in\inputspace$ with $\succsim:\Sigma_{1}\succ\cdots\succ\Sigma_{l}$, as  

$$F^*(\succsim)  =
\begin{cases}
\bigcap \Sigma_{1} &\text{if }\bigcap \Sigma_1\neq \emptyset,\\
F^{P}(\succsim) &\text{if }\bigcap \Sigma_{1}=\emptyset \text{ and }l|\Sigma_1|\text{ is even},\\
X &\text{otherwise.}
\end{cases}
$$

\begin{proposition}
The three axioms in Theorem \ref{theorem: characterization of plurality by STAG}
 are considered logically independent. In fact, $F^*$ satisfies all the axioms except DMON, $F^{splitP}$ satisfies all the axioms except SI, and $F^L$ satisfies all the axioms except STAG. \end{proposition}

\begin{proof}\ 

\vspace{.5\baselineskip}
\noindent
\textbf{\textit{On DMON}}. The weak order $F^*$ satisfies STAG by definition, and given any $\Gamma\in\Sigma_{k}$ with $|\Gamma[x]|=|\Gamma[y]|$ then the parities of $\Sigma_1$, $\Sigma_1\cup\Gamma$, and $\Sigma_1\setminus\Gamma$ are the same, and  $F^*$ satisfies the SI. Consider now the following weak orders $\succsim: \Sigma_1=\{\{x\},\{y\},\{x,y\}\}\succ \mathfrak{X}\setminus\Sigma_1 $ and  $\succsim': \Sigma_1'=\{\{x\},\{y\},\{x,y\}\}\succ' \mathfrak{X}\setminus\Sigma_1'$. We observe that $\succsim'$ is the $\{y\}$-deterioration of $\succsim$ with $\{y\}\in X\setminus\{z\}$ for $z\neq x,y$. Moreover, if $F^*(\succsim)=X$, $F^*(\succsim')=\{x\}$ and $z\notin F^*(\succsim')$, then $F^*$ does not satisfy the DMON. 

\vspace{.5\baselineskip}
\noindent
\textbf{\textit{On SI}}.
For any $\succsim:\Sigma_{1}\succ\cdots\succ\Sigma_{l}$, observe that  
$$\sum_{\substack{S\subseteq \mathfrak{X}[x]\\ S\in\Sigma_1}} \frac{1}{\left| S \right|}\leq \sum_{ S\in\Sigma_1} \frac{1}{\left| S \right|},$$ 
and the equality occurs if and only if $x\in\cup\Sigma_1$, then $F^{splitP}$ satisfies STAG. Moreover, observe that, given $S\subseteq X\setminus\{x\}$, an $S$-deterioration $\succsim'$ of $\succsim$ does not affect the sum $\sum_{\substack{S\subseteq \mathfrak{X}[x]}}$ then if $x\in F^{splitP}(\succsim)$ then $x\in F^{splitP}(\succsim')$ and $F^{splitP}$ satisfies DMON. Take the following weak order $\succsim: \Sigma_1=\{\{x\},\{y\},\{y,z\},\{x,w\}\}$ $\succ \mathfrak{X}\setminus\Sigma_1 $ and define $\Gamma =\{\{x,w\},\{y\}\}$. Observe that $|\Gamma[x]|=|\Gamma[y]|$, but $\tilde{x}_1=1>\tilde{y}_1=\frac{1}{2}$ then $F^{splitP}(\succsim)\cap\{x,y\}\neq F^{splitP}(\succsim')\cap\{x,y\}$ and $F^{splitP}$ does not satisfy SI. 

\vspace{.5\baselineskip}
\noindent
\textbf{\textit{On STAG}}.
Consider the following weak order $\succsim: \Sigma_1=\{\{x,y,z\},\{x,y,w\}\}\succ\Sigma_2=\{\{x\},\{z\}\}\succ \mathfrak{X}\setminus\Sigma_1\cup\Sigma_2$. We have $\cap\Sigma_1=\{x,y\}$, but $F^L(\succsim)=\{x\}$ then, $F^L$ does not satisfy the STAG. If $x\in X$ and $S\subseteq X\setminus\{x\}$, observe that any $S$-deterioration $\succsim'$ of $\succsim$ does not change the coordinates of the vector $\theta(x)$. If $x\in F^L(\succsim)$, then $x\in F^L(\succsim')$ and $F^L$ satisfies DMON. Now, considering $x,y\in X$ and a weak order $\succsim$ such that $F^L(\succsim)\cap\{x,y\}=\{x\}$, there exists $\bar{k}\in[l]$ such that $x_{\bar{k}}>y_{\bar{k}}$ and $x_{k}=y_{k}$ for all $k<\bar{k}$. Let $\Gamma\subseteq X$ with $|\Gamma[x]|=|\Gamma[y]|$ and let $\succsim':\Sigma_{1}'\succ'\cdots\succ'\Sigma_{l}'$, where $\Sigma_{k_{1}}'=\Sigma_{k_{1}}\setminus\Gamma$, $\Sigma_{k_2}'=\Sigma_{k_2}\cup\Gamma$, $\Sigma_k'=\Sigma_k$ for all $k\neq k_{1},k_{2}$  and  $F^L(\succsim')\cap\{x,y\}\neq\emptyset$. Observe that $x_{k_1}'=x_{k_1}-\gamma$, $y_{k_1}'=y_{k_1}-\gamma$ and $x_{k_2}'=x_{k_2}+\gamma$, $y_{k_2}'=y_{k_2}+\gamma$ then again $x_{\bar{k}}'>y_{\bar{k}}'$ and $x_{k}'=y_{k}'$ for all $k<\bar{k}$ and $F^L(\succsim')=\{x\}$. The same argument holds for the other cases when $F^L(\succsim)\cap\{x,y\}\neq\emptyset$ then, $F^L$ satisfies the SI.

\end{proof}

We have shown that STAG (i.e., a strong version of Mill's first canon), DMON, and SI characterize the plurality $F^{P}$. We cannot replace STAG with TAG in Theorem \ref{theorem: characterization of plurality by STAG}. In fact, it is straightforward to confirm that $F^{L}$ satisfies the TAG, DMON, and SI. This indicates that plurality $F^{P}$ and $F^{L}$ both satisfy Mill's first canon (in the absolute top interpretation). 

Nevertheless, the compatibility of $F^{P}$ and Mill's canons is complete. In fact, we can see that among $F^{P}$, $F^{L}$, and $F^{OBI}$, only plurality $F^{P}$ satisfies all Mill's canons (except the fourth canon). The following proposition formally states the following. 

\begin{proposition}
\label{proposition:Mill's canons and SRSs}
Among STAG (i.e., a strong version of Mill's first canon), TAG (i.e., Mill's first canon), TDF (i.e., Mill's second canon), TJAD (i.e., Mill's third canon), and CV ( Mill's fifth canon), (i) $F^{P}$ satisfies all conditions; (ii) $F^{L}$ satisfies only TAG, TDF and TJAD; (iii) $F^{OBI}$ satisfies only TDF  and TJAD.
\end{proposition}

\begin{proof}
(i) STAG has already been shown in Theorem \ref{theorem: characterization of plurality by STAG}. Because TDF and TJAD follow from STAG, we also know that $F^{P}$ satisfies them. Finally, we show that $F^{P}$ satisfies the CV.  Take $\succsim:\Sigma_{1}\succ\cdots\succ\Sigma_{l}$, and assume $x\in C_{\succsim}$. By definition, $S\cup\{x\}\succ S$ for all $S\subseteq X\setminus\{x\}$. This implies that, for any $S\subseteq X\setminus\{x\}$, $S$ does not belong to $\Sigma_{1}$. Thus, $\Sigma_{1}=\Sigma_{1}[x]$. Hence, $x\in F^{P}(\succsim)$. 

(ii) It is straightforward to confirm that $F^{L}$ satisfies the three axioms: TAG, TDF, and TJAD.
$F^{L}$ does not satisfy  STAG, indeed take $\succsim:\{\{x,y\}\}\succ\{\{x\}\}\succ\{\{x,y\},\{x\}\}^c$.  STAG requires that $x,y$ are selected at $\succsim$, but $F^{L}=\{x\}$. 
$F^{L}$ does not satisfy  CV, indeed take $\succsim:\{\{x,y\}\}\succ\{\{y\}\}\succ\mathfrak{X}[x]\setminus\{\{x,y\}\}\succ\left(\mathfrak{X}[x]\cup\{\{y\}\}\right)^c$. CV demands that $x$ be the unique winner but $F^{L}(\succsim)=\{y\}$. 

(iii) It is straightforward to confirm that $F^{OBI}$ satisfies TDF and CV. 
$F^{OBI}$ satisfies TJAD, indeed let $\succsim\in\domain$ with $\succsim:\Sigma_{1}\succ\cdots\succ\Sigma_{l}$ and $x\in X$. Suppose that (a) $\bigcap \Sigma_{1}=\{x\}$, (b) $\bigcap (\Sigma_{2}\cup\cdots\cup\Sigma_{l})=\emptyset$, and (c) $\bigcup (\Sigma_{2}\cup\cdots\cup\Sigma_{l})\subseteq X\setminus\{x\}$. Subsequently, from (c), we obtain $\mathfrak{X}[x]\subseteq \Sigma_{1}$. From (a), we obtain that $\mathfrak{X}[x]^c\cap \Sigma_{1}=\emptyset$. This implies that $x\in C_{\succsim}$, and hence, $x\in F^{OBI}(\succsim)$.
$F^{OBI}$ does not satisfy the TAG, indeed take $y\in X$ with $y\neq x$ and $\succsim:\{\{x\}\}\succ\mathfrak{X}[y]\setminus\mathfrak{X}[x]\succ\mathfrak{X}[y]\cap\mathfrak{X}[x]$,
$\succ\left(\{x\}\cup\mathfrak{X}[y]\right)^c$. Then, TAG demands that $x$ be the unique winner, but $F^{OBI}=\{y\}$. Since STAG is logically stronger than TAG, this also proves that $F^{OBI}$ does not satisfy STAG.
\end{proof}

Finally, we briefly comment on the skipped canon and Mill's fourth canon. 

\begin{quote}
\textbf{\textit{Fourth canon}} Subduct from any phenomenon such part as is known by previous inductions to be the effect of certain antecedents, and the residue of the phenomenon is the effect of the remaining antecedents. 
\end{quote}

The fourth canon is applied when the phenomenon is divided into parts and when the causes of some parts are already known. For example, when the phenomena $P_1,P_2,\text{and }P_3$ under circumstances $C_1,C_2,\text{and }C_3$ are observed, and when we already know by past induction that the causes of $P_1$ and $P_2$ are $C_1$ and $C_2$, then the fourth canon indicates that the cause of $P_3$ is $C_3$. Thus, the fourth canon is applied only when there are multiple types of phenomena, which is why we excluded this canon from this paper. One possible direction for incorporating this canon in the context of SRSs is to consider multiple indicators of the results (such a context has been partly studied in \cite{Suzuki2021}). The evaluation of the cooperative results often involves multiple viewpoints (time, cost, etc.). High-dimensional SRSs represent an interesting topic for future studies. 

\section{Conclusion}
\label{section:conclusion}

This study reinterprets Mill's canons in the context of SRSs, where the main objective is to find the key individuals who play an essential role in cooperative activities. Our results indicate that plurality $F^{P}$, which selects the individuals (circumstances) that appear most often in the top-ranked coalitions, satisfies Mill's canons (except for the fourth canon, which we deal with in a different context) simultaneously. This is remarkable because well-known CSCFs such as LES $F^{L}$ and OBI $F^{OBI}$ fail to satisfy some canons (Proposition \ref{proposition:Mill's canons and SRSs}). In this sense, our results indicate the hidden advantage of plurality in terms of classic causal inference theory. 

\section*{Acknowledgments}
Takahiro Suzuki received a JSPS KAKENHI Grant.
Number JP21K14222 and JP23K22831. Stefano Moretti acknowledges financial support from the ANR project THEMIS (ANR-20-CE23-0018). Michele Aleandri is member of GNAMPA of the Istituto Nazionale di Alta Matematica (INdAM) and is supported by Project of Significant National Interest – PRIN 2022 of title “Impact of the Human Activities on the Environment and Economic Decision Making in a Heterogeneous Setting: Mathematical Models and Policy Implications”- Codice Cineca: 20223PNJ8K- CUP I53D23004320008.

\bibliographystyle{elsarticle-num}

\begin{thebibliography}{10}

\bibitem{Aleandri2024}
Michele Aleandri, Felix Fritz, and Stefano Moretti.
\newblock Desirability and social ranking.
\newblock {\em Social Choice and Welfare}, pages 1--43, 2025.

\bibitem{Algaba2021}
Encarnaci{\'{o}}n Algaba, Stefano Moretti, Eric R{\'{e}}mila, and Philippe Solal.
\newblock {Lexicographic solutions for coalitional rankings}.
\newblock {\em Social Choice and Welfare}, 57(4):817--849, 2021.

\bibitem{Allouche2021}
Tahar Allouche, Bruno Escoffier, Stefano Moretti, and Meltem {\"{O}}zt{\"{u}}rk.
\newblock {Social ranking manipulability for the CP-majority, Banzhaf and lexicographic excellence solutions}.
\newblock In {\em Proceedings of the Twenty-Ninth International Joint Conference on Artificial Intelligence (IJCAI-20)}, pages 17--23, 2021.

\bibitem{Bernardi2019}
Giulia Bernardi, Roberto Lucchetti, and Stefano Moretti.
\newblock {Ranking objects from a preference relation over their subsets}.
\newblock {\em Social Choice and Welfare}, 52(4):589--606, 2019.

\bibitem{Brandt2016a}
Felix Brandt, Vincent Conitzer, Ulle Endriss, J{\'{e}}r{\^{o}}me Lang, and Ariel~D. Procaccia.
\newblock {\em {Handbook of computational social choice}}.
\newblock 2016.

\bibitem{Ducheyne2008a}
Steffen Ducheyne.
\newblock {J.S. Mill's canons of induction: From true causes to provisional ones}.
\newblock {\em History and Philosophy of Logic}, 29(4):361--376, 2008.

\bibitem{Fayard2018a}
Nicolas Fayard and Meltem~Escoffier Ozturk.
\newblock {Ordinal Social ranking : simulation for CP-majority rule}.
\newblock In {\em DA2PL'2018 (From Multiple Criteria Decision Aid to Preference Learning)}, Poznan, Poland. hal-02164682, 2018.

\bibitem{ferey2016multiple}
Samuel Ferey and Pierre Dehez.
\newblock Multiple causation, apportionment, and the shapley value.
\newblock {\em The Journal of Legal Studies}, 45(1):143--171, 2016.

\bibitem{ferey2016overdetermined}
Samuel Ferey and Pierre Dehez.
\newblock Overdetermined causation cases, contribution and the shapely value.
\newblock {\em Chi.-Kent L. Rev.}, 91:637, 2016.

\bibitem{Finn2011}
V.~K. Finn.
\newblock {J.S. Mill's inductive methods in artificial intelligence systems. Part I}.
\newblock {\em Scientific and Technical Information Processing}, 38(6):385--402, 2011.

\bibitem{Haret2019a}
Adrian Haret, Hossein Khani, Stefano Moretti, and Meltem Ozturk.
\newblock {Ceteris paribus Majority for social ranking}.
\newblock {\em 27th International Joint Conference on Artificial Intelligence (IJCAI-ECAI-18), Jul 2018, Stockholm, Sweden}, pages 303--309, 2019.

\bibitem{J.L.Mackie1967}
{J.L. Mackie}.
\newblock {\em {Mill's methods of induction}}.
\newblock 1967.

\bibitem{Khani2019}
Hossein Khani, Stefano Moretti, and Meltem Ozturk.
\newblock {An Ordinal Banzhaf Index for Social Ranking}.
\newblock In {\em 28th International Joint Conference on Artificial Intelligence (IJCAI 2019), Aug 2019, Macao, China.}, pages 378--384, 2019.

\bibitem{Konieczny2022}
S{\'{e}}bastien Konieczny, Stefano Moretti, Ariane Ravier, and Paolo Viappiani.
\newblock {Selecting the Most Relevant Elements from a Ranking over Sets}.
\newblock {\em Lecture Notes in Computer Science (including subseries Lecture Notes in Artificial Intelligence and Lecture Notes in Bioinformatics)}, 13562 LNAI:172--185, 2022.

\bibitem{Mill1882}
John~Stuart Mill.
\newblock {\em {A system of logic, ratiocinative and inductive, being a connected view of the principles of evidence, and the methods of scientific investigation}}.
\newblock Harper \& Brothers, Publishers, Franklin Square [Project Gutenberg, Ebook-No.27942], 1882.

\bibitem{Moretti2017}
Stefano Moretti and Meltem {\"{O}}zt{\"{u}}rk.
\newblock {Some axiomatic and algorithmic perspectives on the social ranking problem}.
\newblock In J.~Rothe, editor, {\em Algorithmic Decision Theory. ADT 2017. Lecture Notes in Computer Science (including subseries Lecture Notes in Artificial Intelligence and Lecture Notes in Bioinformatics)}, volume 10576 LNAI, pages 166--181. Springer, Cham, 2017.

\bibitem{Pietka2015}
Dariusz Pi{\c{e}}tka and Pawe{\l} Stacewicz.
\newblock {A decision logic approach to Mill's eliminative induction}.
\newblock {\em Studies in Logic, Grammar and Rhetoric}, 42(55):113--138, 2015.

\bibitem{Prueitt1998}
Paul~S. Prueitt.
\newblock {Interpretation of the logic of J.S. Mill}.
\newblock {\em IEEE International Symposium on Intelligent Control - Proceedings}, pages 820--827, 1998.

\bibitem{Suzuki2021}
Takahiro Suzuki and Masahide Horita.
\newblock {Social Ranking Problem Based on Rankings of Restricted Coalitions}.
\newblock In D.C. Morais, L.~Fang, and M.~Horita, editors, {\em Contemporary Issues in Group Decision and Negotiation. GDN 2021. Lecture Notes in Business Information Processing, vol.420}, volume 420, pages 55--67. Springer, Cham, 2021.

\bibitem{Suzuki2023_Mill}
Takahiro Suzuki and Masahide Horita.
\newblock Revisiting js mill’s methods: Causal inference and social choice theory.
\newblock {\em Available at SSRN 4603153}, 2023.

\bibitem{Suzuki2024}
Takahiro Suzuki and Masahide Horita.
\newblock {Consistent Social Ranking Solutions}.
\newblock {\em Social Choice and Welfare}, 62:549--569, 2024.

\bibitem{Suzuki2024c}
Takahiro Suzuki and Masahide Horita.
\newblock Sabotage-proof social ranking solutions.
\newblock {\em Theory and Decision}, pages 1--20, 2024.

\bibitem{Suzuki2024a}
Takahiro Suzuki, Yu~Maemura, and Masahide Horita.
\newblock A unified understanding of majority rule, cp majority rule, and their variants.
\newblock {\em CP Majority Rule, and Their Variants.(January 23, 2024)}, 2024.

\bibitem{Taylor2002}
Alan~D. Taylor.
\newblock {The Manipulability of Voting Systems}.
\newblock {\em The American Mathematical Monthly}, 109(4):321--337, 2002.

\end{thebibliography}

\end{document}